\documentclass[reqno,12pt]{article}

\usepackage{amsmath,amsthm}
\newcommand{\wh}{\widehat}
\newcommand{\wt}{\widetilde}
\newcommand{\bfE}{{\bf E}}
\newcommand{\bfe}{{\bf e}}

\newcommand{\bfg}{{\bf g}}

\newcommand{\bR}{{\bf R}}
\newcommand{\pa}{\partial}

\newcommand{\la}{\lambda}

\newcommand{\str}{\stackrel}

\newcommand{\cL}{{\mathcal L}}

\newcommand{\cH}{{\mathcal H}}
\newcommand{\cE}{{\mathcal E}}

\def\Symm{\operatorname{Symm}}
\def\eqref#1{(\ref{#1})}

\newtheorem*{theorem}{Theorem}

\newtheorem{proposition}{Proposition}

\theoremstyle{definition}
\newtheorem{definition}{Definition}

\theoremstyle{remark}
\newtheorem{remark}{Remark}
\begin{document}

\title{Economic law of increase of Kolmogorov complexity.
       Transition from financial crisis
       2008 to the zero-order phase transition
      (social explosion)}

\author{V.P.~Maslov}

\date{}
\maketitle

\begin{abstract}
In \cite{MMM}, a two level model of the occurrence of financial
pyramid (bubbles) has been considered. We also considered the
mathematical analogy  of this model to Bose condensation. In the
present paper, we explain why Ponzi schemes and bubbles result in
a crisis in real economics. In \cite{DAN_Port}, the law of
increase of entropy in financial systems, and consequently
increase of Kolmogorov complexity, is formulated. If this law is
broken, the financial system  makes a phase transition to a
different state. In \cite{Prom_Math} the author considered a two
level model of the zeroth-order phase transition which was
interpreted in \cite{QuantumEc} as an analog of social
catastrophe. In the present paper we also examine this model.
\end{abstract}

By the end of 2008, the United States has accumulated considerable
debts. The debts can be counted by the same scheme as GDP, that is,
as the difference of non-repayed and repayed debts in a year
provided that this difference is not zero, i.e., the debt has
increased owing to the velocity of debt.

Just as negative money, i.e., debts, do not occur in the definition of
velocity of money and in the calculation of the GDP, so the
definition of the velocity of debt does not involve positive money
(earned money, etc.).

If $N_i$ is the debt resulting from one turnover (one act of
borrowing from somebody and paying the dept to somebody else),
then the total debt is equal to $M_i=N_i\sigma$, where
$\sigma$, the number of turnovers per annum, will be called the
velocity of debt (by analogy with ``velocity of money''). The
highest velocity of debt is observed for bubbles, or Ponzi
schemes. A mortgage debt has only one turnover. For a car loan,
the velocity of debt is also equal to one if the customer made a
large borrowing and has not repaid the loan yet.

Thus, $\sum N_i\sigma_i \leq M$, where $M$ is the total debt (like
a ``negative GDP'').
Hence the debts obey the same theorems of number theory as GDP
(see \cite{QuantumEc} and also \cite{Negat_Dim}).

If $M$ is the total debt and $\sigma$ is the average velocity of
debt (analog of the temperature $\theta$ in physics), then the
``negative money supply'' is given by $N=M/\sigma$, where
$N=\sum_{i=1}^k N_i$, $k\gg 1$, and $N_i$ corresponds to the
$i$the set of subjects, which has the cardinality $\alpha_i$.
If we adopt, similarly to the Baumol--Tobin model,
that the velocity of debt of a single subject is proportional
to the square root of its debt (cf.~\cite[Appendix to Chapter~13]{BT})
and take into account that the debt is equal to $iN_i$,
then $N_i$ is proportional to $i$ and it is more accurate to assume that
$N_i=in_i$, where $n_i$ are distributed uniformly.
Hence $\sum^{k}_{i=1}\alpha_i i^2 n_i=M$ and
$\sum in_i=N$. For $\alpha_i$ we assume that they decrease
according to the Pareto distribution.

The critical number $N_0$ after which the ''Bose condensation
phenomenon'' occurs is computed in~\cite{QuantumEc,Masl_Naz_1_III,NegDim}.
In these papers, the Pareto distribution was proved
considered as negative dimension,
namely,
it was proved that the larger the Pareto index
for the distribution of $\alpha_i$
and the lesser the Baumol--Tobin dependence,
then the larger the critical number $N_0$ and the more difficult
it can be overcome.

However, there is a substantial difference between inflation
occurring if $N>N_0$ for the case of income and the case in which
$N>N_0$ for debts.

In the first case (inflation) small denominations fall out of
circulation (part of the money is abolished), and the nominal
value of money is reduced. In the second case, debts cannot be
annihilated by denomination, and excessive debts for $N>N_0$
accumulate at the ''lower level'', i.e., at the subjects who have
one turnover per annum. The zero velocity of debt will be taken
into account in incomes rather than debts.
On the other hand, it follows from the theorem proved
in~~\cite[Appendix to Chapter~13]{BT}
that small velocities of both the money and the debts
practically do not affect the value of~$N_0$.
Hence the definition of nonzero lower level of velocity of money
is not essential in the rather schematic GDP model.

However, since the lower velocity level for debts plays an
exceptional role and the model involving the number of turnovers is
too coarse for determining the least velocity of debt, one has to
single out the ``longest-term debts'' whose repayment spans over
many years, i.e., whose repayment velocity is small.

How the long-term project contribution to GDP is calculated?
It is necessary to calculates the excess
of the cost of the project already performed part in a year
over the cost of the entire project in $L$ years
minus the total expenditures in $L$ years divided by~$L$.

Similarly, the rate of the debt ``turnover''
in the case of a long-term credit obtained for $L$ years
is equal to $1/L$.
Therefore, if, without loss of generality,
we assume that the numbers of turnovers larger than~$2$
are integer numbers,
then, for a single ``turnover'' (a credit taken once),
we must take the credit period (the number of years $L$) into account.
Correspondingly, the rate of the debt turnover equal to~$1$
splits into rates of the debt turnover that are less than~$1$.
Hence the debt Bose condensation accumulates precisely
at these slow rates of the debt turnovers,
and the bankruptcy expects precisely them.
This follows from the mathematical theorem of the number theory.
But the theorem does not speak how this is realized in the society.
Possibly, to escape the bankruptcy,
the debtors with high level of the debt velocity
must again take long-term deposits.
Perhaps, this chain can be explained by economists,
but the mathematical law inexorably predicts this scenario
of the debt crisis.

It is on these debts that the excessive debts exceeding $N_0$
mainly fall; as a result, the long-term borrowers default on
their debts, which fall out of the chain and descend on the
subsequent levels.

Hence mortgage is the first to collapse, and the car loans follow.

Thus, the mathematical theory, which pertains to number theory, is
the same, but the result in economics is different.

It follows from the preceding that to fight the crisis and to
avert it in future one should favor a debt structure in which the
share of high-velocity debts is smaller that in the existing
structure.
In~\cite{QuantumEc,Masl_Naz_1_III}, it was proved that if the Pareto index
is not too large, then it is possible to increase the critical number
by introducing a larger number of currencies, for example,
in each of the United States,
by introducing its own currency in addition to the dollar.

\newpage

\section{Economic Law  of Increase of Kolmogorov  Complexity}

In his celebrated paper on complexity, A.N. Kolmogorov defined
Kolmogorov complexity as a minimal code of a sequence of numbers.
He also mentioned some other complexity, which was related to
decoding. For the stock exchange, the definition of a code length
that involves decoding is more appropriate to the situation under
consideration, which leads to an arbitration-free situation in
the end.

Suppose that a trader contrived a combination that brought a big
gain to this trader. The market intuitively deciphers this
combination (algorithm) and, thereby, annihilates it. Then, the
trader invents some more complicated algorithm, and so on. In the
end, the complexity of the algorithm becomes so high that it only
slightly differs from a pseudorandom complexity. In the limit,
this leads to an arbitration-free situation.

For this reason, it is more natural to understand complexity as
the length of a code with polynomial decoding complexity. This
code should be used to encode a sequence. The longer the code,
the closer the situation under consideration to general position,
which corresponds to the law of large numbers. If the minimum
specified above tends to zero as $N\to\infty$, then the set of
deals $\{N_i\}$ is close to an arbitration-free set.

In the language of probability theory, the initial probabilities
corresponding to $\lambda_i$, which equal $G_i/G$ in the sense of
von Mises, tend to zero as $N\to\infty$ (and, hence, as
$G\to\infty$).

Therefore, the initial mathematical expectation
$$
\frac{\sum g_i \la_i}{\sum g_i}= \sum\frac{g_i \la_i}{G}
$$
is such that each term of the probability $p_i$ vanishes in the
limit. However, the mathematical expectation tends to a constant
larger than zero and the partial sums
$\sum_{i=i_\alpha}^{j_\alpha} p_i$ do not tend to zero.

Similarly, the mathematical expectation $\sum\frac{N_i \la_i}{N}$
does not tend to zero, whereas  $\frac{N_i}{N} \to 0$, but
$\frac1N \sum_{i=i_\alpha}^{j_\alpha} N_i$ does not tend to zero.

This is a very important economy factor, especially in the
psychological aspect (see \cite{QuantumEc,Luce}).

Since the tendency of the market to an arbitration-free situation
is equivalent to the eagerness of the market participants to gain
profits, it follows that the entropy, if it depends on some
parameters, must increase with respect to these parameters most
rapidly, in the direction of steepest ascent.

If $H(E, N)$ depends on $k$ additional $x_1, \cdots, x_k$, i.e.,
$H=H(x_1, \cdots, x_k)$, (the dependence on $E$ and $N$ is not
indicated), then
$$
\str{.}{x}_i=c(H)\frac{\pa H}{\pa x_i}
$$
\begin{equation}
\frac{\pa H}{\pa t} =\sum_i^k \frac{\pa H}{\pa x_i}\frac{\pa
x_i}{\pa t}=c(H)\sum_i^k \bigg(\frac{\pa H}{\pa x_i}\bigg)^2, \ \ c(H)>0 \ \
H|_{t=0}=H(x_1, \cdots, x_k). \label{N18'}
\end{equation}

For simplicity, suppose that $C(H) = 1$. Otherwise, setting
$F(H)=\int C(H) dH$, we obtain $\frac{\pa F}{\pa t}=(\nabla F)^2$.
This implies
$$
H(x,t)=\min_\xi\bigg(\frac{(x-\xi)^2}{2t}+H(\xi)\bigg).
$$
\begin{remark}
Performing geometric quantization \cite{Quantization}, we obtain
$$
\wh{H} (x,t)= \ln\frac1{\sqrt{t^k}}\int e^{-\frac{(x-\xi)^2}{2t}}
e^{H(\xi)}d\xi, \ \ \xi \in R^k.
$$
This asymptotically coincides with the tunnel canonical operator
on the Lagrangian manifold  $p_i=\frac{\pa M}{\pa x_i}$ shifted
along the trajectories  $\str{.}{x}=p$ and $\str{.}{p}=0$. The
quantization of Eq. \eqref{N18'}  leads to solution of the heat
equation \cite{Shiryaev-Fund}
\begin{equation}
N\frac{\pa u}{\pa t}=\cH (\ln u^2) \Delta u \qquad  u|_{t=0}
=H(x_1, \cdots, x_n), \label{N19'}
\end{equation}
Quantizing this equation (as a Bose system) and setting
$u=\Psi^+$ and $ u^2=\Psi \cdot \Psi^-$, we obtain a linear
equation whose asymptotics is again the canonical operator.

The prices $\la_i(x_1, \cdots, x_k)$  vary according to the
equation
\begin{equation}
\frac{d\la_i}{dt}= \sum_j\frac{\pa \la_i}{\pa x_j}\frac{dx_i}{dt}
= \sum_j\frac{\pa \la_i(x_1, \cdots, x_k)}{\pa x_j} C(H)\frac{\pa
H}{\pa x_J}. \label{N20'}
\end{equation}
\end{remark}

Thus, if we know the change of some (specific) price, we can find
the function $C(H)$ and, thereby, solve the problem completely.

A similar situation may arise when, e.g., barriers between
countries (such as customs tariffs) decrease and prices tend to even
according to the law of steepest ascent of entropy, because
profiteers begin to use the arising arbitration (difference in
prices) at a furious pace, which leads to an arbitration-free
situation most rapidly. Such was the situation in the Soviet
Union in the late 1980s when the iron curtain collapsed
\cite{QuantumEc}.

On September 30, 2008, the National Debt Clock in Manhattan ran
out of digits as the United States public debt exceeded \$10
trillion, a significant symbol of the current financial and
economical crisis. There are many diverse factors behind the
crisis, and they are naturally a subject of broad interest.
However, apart from economics laws, which are widely discussed
nowadays, there is a purely mathematical law contributing to the
disaster, an inexorable law of numbers, which economists, let
alone the general public, fail to recognize.

Suppose we want to deposit two kopecks in two different banks
(see \cite{QuantumEc}; then we can say that there are three
possibilities: (1) put both kopecks into one of the banks; (2)
put both kopecks into the other bank; (3) one kopeck in one bank,
the other, in the other one. Here it is of no consequence which
of the two coins we deposit in the first bank and what is its
year of issue. Now imagine a situation in which we are depositing
one kopeck and one pence instead of two kopecks. In that case, we
have four options rather than three, because it is significant
which coin we placed in what bank, and so the variant in which
the coins are placed in different banks yields two different
options: (a) one kopeck to bank 1 and one pence to bank 2; (b)
one kopeck to bank 2 and one pence to bank 1.

These two cases give rise to two different laws in number theory
(see \cite{NumberTheory1}). In the second case, the corresponding
law implies that the total number of cents (and hence of dollars
and of millions of dollars), i.e., the total amount of money in
circulation, cannot be arbitrarily large. There is a threshold,
and if the amount of money in a given currency exceeds that
threshold, then an economic disaster occurs. The same thing
pertains to negative amounts of money, i.e., debts.

Mathematically, this kind of disaster can be described as an
analog of ``Bose condensation'': should there be too many
particles in a system, all the excessive particles would collapse
into the ground state. The condensation effect is indeed observed
in finance: for example, if the inflation rate is too high (too
much money has been issued), then the lower denominations (like
cents, pence, or kopecks) die away, i.e., are withdrawn from
circulation.

It is important to note that the condensation phenomenon is solely
due to the fact that particles are indistinguishable---there is no
condensation at all if the particles are distinguishable (and so
obey the Gibbs statistics). One should simultaneously use as many
currencies as possible. For example, if, in addition to dollar,
its own currency were introduced in each of the United States,
then the sum of the thresholds for all states would be much
greater than the threshold in the case of a single currency; it
would be much harder to exceed this new threshold, and hence the
disaster would at least be postponed. By the same pattern, all
national currencies should have been retained along with the
euro. Then the violation of a country's threshold
could cause a crisis in that country but would not affect the
other economies, just as the 1998 default in Russia did not bring
down the world economy.
(Unfortunately, the dollar largely circulates outside the United States,
and hence the world on the whole is much more affected this time).

Thus, the introduction of $K$ currencies has raised $N_{cr}$
by the factor of approximately $K^{1/2}$.
Or, the introduction of concurrently existing $K$ currencies increases the crisis
threshold by the factor $K^{1/2}$.
The outcome will be more or less the same regardless of whether these new
currencies are nationwide or function only within their respective states.

The above  described transition from   the MMM pyramid
(''bubbles'') leads to entropy increase, therefore, to the
Kolmogorov complexity. Hence  ''bubbling'' leads to decrease of
complexity. Likewise, transition to only one currency leads to
decrease of complexity, and this  is contributing to a crisis.
Development of advanced technologies  results in increase of
complexity, therefore it hinders a crisis.

Thus, a general law of economics  is  the law of \textit{
increase of complexity}\footnote{The analogy to the
above-mentioned effect in statistical physics is given  e.g., in
\cite{Zurek}.}. Otherwise, a crisis is imminent.

\section{Nonstationary financial averaging}

Consider a random variable $\lambda$ which takes $l$ different
values $\lambda_1,\dots,\lambda_l$. In what follows, we assume
that the values of this random variable satisfy the following
condition: if $k_1,\dots,k_l$ are integers such that
\begin{equation}
\sum_{i=1}^lk_i=0,\qquad\sum_{i=1}^l\lambda_i k_i=0. \label{ev1}
\end{equation}
then $k_i = 0$  for all $i = 1, 2, ..., l$.

In \cite{Axiom}, the notion of financial averaging and financial
averaging axioms were introduced. Suppose given a convex function
$f(x)$ having inverse  $f^{-1}(x)$, where $x\in\bR$ and a set of
nonnegative numbers $p_i$, where $i=1,\dots,l$; we call them
weights in what follows.

The first axiom \cite{Axiom} is that the financial averaging of
the random variable $\lambda$ corresponding to the function
$f(x)$ and the set of weights$p_i$, $i=1,\dots,l$ is the
expression
\begin{equation}
\bfE_{f,p}(\lambda)=f^{-1}\left(\sum_{i=1}^l p_i
f(\lambda_i)\right). \label{ev3}
\end{equation}

We modify the fourth axiom as follows: For any set of weights
$p_i$, any set of numbers  $\lambda_i$ , and any number $a$, the
following equality holds:
\begin{equation}
\bfE_{f,p}(\wt{\lambda})=Ca+\bfE_{f,p}(\lambda), \label{ev4}
\end{equation}
where $\wt{\lambda}_i=\lambda_i+a$, and $C$ is a number not
depending on $a$. Then, the following assertion is valid.

\begin{proposition} Condition \eqref{ev4} implies that the function $f(x)$ has the form
\begin{equation}
f(x)=A\exp(-\beta x),\qquad \text{or}\qquad f(x)=A x+D,
\label{ev5}
\end{equation}
where the numbers $\beta$ and  $A$ are nonzero.
\end{proposition}

This assertion is proved by the method described in \cite{Axiom}.

Hereafter, we consider the case of a nonlinear function $f(x)$.
According to Propositon 1, the financial averaging of a random
variable $\lambda$ corresponding to a set of weights $p_i$ has the
form
\begin{equation}
\bfE_{p}(\lambda)=-\frac1\beta\ln\left(\sum_{i=1}^l p_i
e^{-\beta\lambda_i} \right).\label{ev6}
\end{equation}

Properly, Axioms 2 and 3 \cite{Axiom} do not hold; thus, we
regard this averaging as nonstationary and consider the process
stabilizing it. The stabilization procedure does not essentially
use the fourth axiom.

Following \cite{RJMP}, we associate the random variable under
consideration with an n-dimensional space $L$ and an operator
$\Lambda$ in this space. The operator $\Lambda$ has eigenvectors
$\bf{e}_i$, where $i = 1, 2, \dots, l$, which form an orthonormal
basis in the space $L$. The eigenvalue corresponding to the
eigenvector $\bf{e}_i$ equals $\lambda_i$. Consider a series of
$M$ identical trials whose outcomes are occurrences of the values
$\lambda_1,\dots,\lambda_n$  of the random variable. According to
\cite{RJMP}, this probability problem corresponds to the space
$\cL_M$ being the tensor product of $M$ copies of $L$. The
vectors in the space $\cL_M$ have the form
\begin{equation}
\Psi=\sum_{i_1=1}^l\dots\sum_{i_M=1}^l \psi(i_1,\dots,i_M)
\bfe_{i_1}\otimes\dots\otimes\bfe_{i_M}, \label{ev7}
\end{equation}
where  $\psi(i_1,\dots,i_M)$ is an arbitrary function of the set
of discrete variables   $i_s=1,\dots,l$, $s=1,\dots, M$.

To the series of trials under consideration, there corresponds,
in addition to the space $\cL_M$, the operator $\wh{\cH}_M$
acting on vectors \eqref{ev7} as
\begin{equation}
\wh{\cH}_M\Psi=\sum_{i_1=1}^l\dots\sum_{i_M=1}^l
\left(\sum_{s=1}^M\lambda_{i_s}\right)
\psi(i_1,\dots,i_M)\bfe_{i_1}\otimes\dots \otimes\bfe_{i_M}.
\label{ev8}
\end{equation}

In further considerations, we use the notation
\begin{equation}
M_i(i_1,\dots,i_M)=\sum_{s=1}^M\delta_{ii_s}, \label{ev9}
\end{equation}
where  $\delta_{ii'}$ is the Kronecker symbol. Let us define
projectors $\wh{P}_{\{M\}}$
\begin{equation}
\wh{P}_{\{M\}}\Psi=\sum_{i_1=1}^l\dots\sum_{i_M=1}^l\left(\prod_{i=1}^l
\delta_{M_iM_i(i_1,\dots,i_M)}\right)
\psi(i_1,\dots,i_M)\bfe_{i_1}\otimes\dots \otimes\bfe_{i_M}.
\label{ev10}
\end{equation}
Here, $\{M\}$ denotes an arbitrary set of nonnegative integers
$M_1,\dots,M_n$ satisfying the condition
\begin{equation}
\sum_{i=1}^lM_i=M. \label{ev11}
\end{equation}

It is easy to show that operators \eqref{ev11}  satisfy the
equality
\begin{equation}
\wh{P}_{\{M\}}\wh{P}_{\{M'\}}=\prod_{i=1}^l\delta_{M_iM_i'}
\wh{P}_{\{M\}} \label{D12}
\end{equation}
thus, these operators are indeed projectors. The operator
$\wh{\cH}_M$  defined by ~\eqref{ev8} is expressed in terms of
the projectors $\wh{P}_{\{M\}}$ as follows:
\begin{equation}
\wh{\cH}_M=\sum_{\{M\}}\left(\sum_{i=1}^l\lambda_iM_i\right)\wh{P}_{\{M\}}.
\label{D13}
\end{equation}
In the subspace  $\cL_{\{M\}}$, onto which the space $\cL_M$  is
projected by $\wh{P}_{\{M\}}$,  we distinguish the basis formed by
vectors of the form
\begin{equation}
\bfE_{i_1,\dots,i_M}=\bfe_{i_1}\otimes\dots\otimes\bfe_{i_M},
\label{D14}
\end{equation}
where $\bfe_i$ occurs $M_i$ times, or, in other words, the set
$i_1,\dots,i_M$ satisfies the conditions
\begin{equation}
M_i(i_1,\dots,i_M)=M_i. \label{D15}
\end{equation}
Consider the following system of vectors $\Phi_{\{M\}}$:
\begin{equation}
\Phi_{\{M\}}=\wh{P}_{\{M\}}\bfe\otimes\dots\otimes\bfe,\qquad
\bfe=\sum_{i=1}^l\bfe_i. \label{ev15}
\end{equation}
Vectors~\eqref{ev15} are equal to the sum of all vectors
~\eqref{D14} from the subspace $\cL_{\{M\}}$; in addition,
\begin{equation}
\Phi_{\{M\}}=\frac1{M!}\left(\prod_{i=1}^lM_i!\right)
\mbox{Symm}_{i_1,\dots,i_M}\bfE_{i_1,\dots,i_M}, \label{D16}
\end{equation}
where $\mbox{Symm}_{i_1,\dots,i_M}$  is the operator of
symmetrization over the indices  $i_1,\dots,i_M$.

We endow the space $\cL_M$  with the norm
\begin{equation}
\|\Psi\|=\sum_{i_1}^l\dots\sum_{i_M=1}^l |\psi(i_1,\dots,i_M)|.
\label{ev12}
\end{equation}

\section{The evolution process}
In this section, we define the notions of evolution and data
reduction.

\begin{definition}The one-step evolution of an ensemble
consisting of $M$ elements, i.e., of a vector $\Psi\in\cL_M$, is
the following vector belonging to the space $\cL_M$:
\begin{equation}
\Psi_1=\exp\left(-\beta\wh{\cH}_M\right)\Psi, \label{ev13}
\end{equation}
where  $\beta$ is a parameter.
\end{definition}
\begin{definition}
The data reduction, or factorization, is the operation that maps
each vector $\Psi\in\cL_M$ to another vector according to the rule
\begin{equation}
R(\Psi)=\sum_{\{M\}}\|\wh{P}_{\{M\}}\Psi\|\Phi_{\{M\}}.
\label{ev14}
\end{equation}
where summation is over all sets of nonnegative integers
$M_1,\dots,M_n$, satisfying condition~\eqref{ev11}.
\end{definition}

To each set of nonnegative numbers $g_i$, where  $i=1,\dots,l$ we
assign the following vector of the space $\cL_M$:
\begin{equation}
\Psi_g=\bfg\otimes\dots\otimes \bfg, \label{ev16}
\end{equation}
where $\bfg$ is the vector from $L$, defined by
\begin{equation}
\bfg=\sum_{i=1}^l g_i\bfe_i. \label{ev17}
\end{equation}

\begin{proposition}
If $g_i=p_ie^{-\beta\lambda_i}$, then the expression
\begin{equation}
-\frac1{M\beta}\ln\left(\|\Psi_g\|\right) \label{ev18}
\end{equation}
coincides with financial averaging~\eqref{ev6}
\end{proposition}
\begin{proof}
Let us write vector~\eqref{ev16} in the form
\begin{equation}
\Psi_g=\left(\sum_{i_1=1}^lg_{i_1}\bfe_{i_1}\right)\otimes\dots\otimes
\left(\sum_{i_M=1}^lg_{i_M}\bfe_{i_M}\right)=\sum_{i_1=1}^l\dots\sum_{i_M=1}^l
g_{i_1}\dots g_{i_M}\bfe_{i_1}\otimes\dots\otimes\bfe_{i_M}.
\label{D19}
\end{equation}
Substitution of \eqref{D19} into~\eqref{ev12}  yields
\begin{equation}
\|\Psi_g\|=\sum_{i_1=1}^l\dots\sum_{i_M=1}^l g_{i_1}\dots
g_{i_M}=\left( \sum_{i=1}^l g_i\right)^M. \label{D20}
\end{equation}
Taking into account that $g_i ==p_ie^{-\beta\lambda_i}$ and
applying~\eqref{D20}, we obtain
\begin{equation}
-\frac1{M\beta}\ln\left(\|\Psi_g\|\right)=-\frac1{\beta}\ln\left(\sum_{i=1}^l
p_ie^{-\beta\lambda_i}\right). \label{D21}
\end{equation}
The right-hand side of~\eqref{D21} coincides with the
right-hand side of \eqref{ev6}, which proves the statement.
\end{proof}

In what follows, we do not use the axiomatics of nonlinear
averaging; instead, we consider the evolution of an arbitrary
(nonnormalized) probability distribution of $g_i$.

For any set of nonnegative $g_i$,  where  $i=1,\dots,l$ we define
the evolution of vectors $\Psi_g(n)\in\cL_{\{M\}}$  by the
recursive formulas
\begin{equation}
\Psi_g(n+1)=R\left(\exp\left(-\beta\wh{\cH}_M\right)
\Psi_g(n)\right),\quad n=0,1,\dots,\quad \Psi_g(0)=\Psi_g.
\label{ev19}
\end{equation}
We use $\psi_g(i_1,\dots,i_M;n)$ to denote the coefficients in
expansion~\eqref{ev7} of the vector $\Psi_g(n)$ in
basis~\eqref{D14}. Thus, $\psi_g(i_1,\dots,i_M;n)$ and
$\Psi_g(n)$. are related by
\begin{equation}
\Psi_g(n)=\sum_{i_1=1}^l\dots\sum_{i_M=1}^l
\psi_g(i_1,\dots,i_M;n)\bfe_{i_1}\otimes\dots\otimes\bfe_{i_M}.
\label{ev20}
\end{equation}
The function $\psi_g(i_1,\dots,i_M;n)$ can be regarded as a
nonnormalized distribution function of the factored ensemble
after $n$ steps under the condition that, at the initial moment,
the distribution function of the factored ensemble has the form
$\psi_g(i_1,\dots,i_M;0)=g_{i_1}\dots g_{i_M}$. In this case,
$g_i$ has the meaning of a nonnormalized distribution function
for one system in the factored ensemble.

Consider the functions
\begin{equation}
F(n,g,M)=-\frac1{M\beta(n+1)} \ln\left(\|\Psi_g(n)\|\right),
\label{ev21}
\end{equation}
and
\begin{equation}
w_i(n,g,M)=\frac1{\|\Psi_g(n)\|}\sum_{i_2=1}^l\dots\sum_{i_M=1}^l
\psi_g(i,i_2,\dots,i_M;n). \label{ev22}
\end{equation}
Since $\psi_g(i_1,\dots,i_M)$ is a distribution function for the
factored ensemble of $M$ systems, formula~\eqref{ev22} defines a
distribution function of one system. Consider the limit of
distribution functions~\eqref{ev21} and~\eqref{ev22} for one
system of the factored ensemble as the number $M$ of its elements
tends to infinity.

\begin{theorem}
As $M\to\infty$  the functions $F(n,g,M)$ and  $w_i(n,g,M)$
approach the following limits: :
\begin{eqnarray}
&&\lim_{M\to\infty} F(n,g,M)=-\frac1\beta\ln\left(\sum_{i=1}^l
e^{-n\beta\lambda_i/(n+1)}g_i^{1/(n+1)}\right)\equiv \wt{F}(n,g),\label{ev23} \\
&&\lim_{M\to\infty}
w_i(n,g,M)=\frac{e^{-n\beta\lambda_i/(n+1)}g_i^{1/(n+1)}}
{\sum_{j=1}^le^{-n\beta\lambda_j/(n+1)}g_j^{1/(n+1)}}\equiv\wt{w}_i(n,g).
\label{ev24}
\end{eqnarray}
\end{theorem}
\begin{proof}
Vector~\eqref{ev16} can be represented in the form
\begin{equation}
\Psi_g=\sum_{\{M\}}\left(\prod_{i=1}^lg_i^{M_i}\right)\Phi_{\{M\}}.
\label{ev25}
\end{equation}
Substituting~\eqref{ev25} into~\eqref{ev19}, we obtain
\begin{equation}
\Psi_g(n)=(M!)^l\sum_{\{M\}}\left(\prod_{i=1}^l
\frac{g_i^{M_i}e^{-n\beta\lambda_i
M_i}}{(M_i!)^n}\right)\Phi_{\{M\}}. \label{ev26}
\end{equation}
This implies
\begin{equation}
\psi_g(i_1,\dots,i_M;n)=(M!)^l
\left(\prod_{i=1}^l\frac{\left(g_ie^{-n\beta\lambda_i}\right)^{M_i}}
{(M_i!)^n}\right), \label{ev27}
\end{equation}
where $M_i$ are expressed in terms of  $i_1,\dots,i_M$ according
to~\eqref{ev9}. Substituting~\eqref{ev27} into~\eqref{ev12}
and summing~\eqref{ev27} over all but one subscripts, we obtain
\begin{eqnarray}
&&\|\Psi_g(n)\|=(M!)^{n+1}\sum_{\{M\}}\prod_{i=1}^l
\frac{g_i^{M_i}e^{-n\beta\lambda_iM_i}}{(M_i!)^{n+1}},\label{ev28}\\
&&\sum_{i_2=1}^l\dots\sum_{i_M=1}^l\psi_g(i,i_2,\dots,i_M;n)=
(M!)^{n+1}\sum_{\{M\}}\left(\prod_{i=1}^l\frac{g_i^{M_i}e^{-n\beta\lambda_iM_i}}
{(M_i!)^{n+1}}\right)\left(\sum_{j=1}^l\frac{M_j}M\right).\label{ev29}
\end{eqnarray}
To determine the asymptotic behavior of the sums in~\eqref{ev28}, \eqref{ev29}
as $M\to\infty$, we apply the Laplace
method. This method gives the equalities
\begin{eqnarray}
&&\lim_{M\to\infty}\frac1M\ln\left(\|\Psi_{p,\lambda}(n)\|\right)=
(n+1)\ln\left(\sum_{i=1}^l e^{-n\beta\lambda_i/(n+1)}
g_i^{1/(n+1)}\right),
\label{ev30} \\
&&\lim_{M\to\infty}\frac1{\|\Psi_{p,\lambda}(n)\|}
\sum_{i_2=1}^l\dots\sum_{i_M=1}^l\psi_{p,\lambda}(i,i_2,\dots,i_M;n)=
\frac{e^{-n\beta\lambda_i/(n+1)}g_i^{1/(n+1)}}
{\sum_{j=1}^le^{-n\beta\lambda_j/(n+1)}g_j^{1/(n+1)}}.
\label{ev31}
\end{eqnarray}
The substitution of~\eqref{ev30} and~\eqref{ev31} into~\eqref{ev21}
and~\eqref{ev22} yields~\eqref{ev23} and~\eqref{ev24}.
This completes the proof of the theorem.
\end{proof}

Thus, by virtue of~\eqref{ev24}, the normalized distribution
$\wt{w}_i(0,g) = \frac{g_i}{\sum_1^e g_i}$ transforms into
$\wt{w}_i(n,g)$ in $n$ evolution steps.

Now, consider the behavior of $\wt{F}(n,g)$~\eqref{ev23} and
$\wt{w}_i(n,g)$~\eqref{ev24} as $n\to\infty$, i.e., the limit
distribution. Let $I$ be a nonempty subset of the set $1,\dots,l$.
Suppose that the $g_i$ satisfy the condition
\begin{equation}
g_i>0,\quad\mbox{for}\quad i\in I,\qquad g_i=0,\quad\mbox{for}\quad
i\notin I. \label{ev32}
\end{equation}
According to ~\eqref{ev23} and \eqref{ev24}, for any $g_i$
satisfying~\eqref{ev32}, we have
\begin{eqnarray}
&&\lim_{n\to\infty}\wt{F}(n,g)=-\frac1\beta\ln\left(\sum_{i\in I}
e^{-\beta\lambda_i}\right),\label{ev33} \\
&&\lim_{n\to\infty}\wt{w}_i(n,g)=\frac{e^{-\beta\lambda_i}}
{\sum_{j=1}^le^{-\beta\lambda_j}}\equiv\rho_i.\label{ev34}
\end{eqnarray}
If the set $I$ coincides with the set of $1,\dots,l$, expressions~\eqref{ev33}
and~\eqref{ev34} are, respectively, the free
energy and Gibbs distribution at temperature $1/\beta$ of a
system with nondegenerate energy levels
$\lambda_1,\dots,\lambda_l$. If $I$ is smaller than  $1,\dots,l$,
then expressions~\eqref{ev33} and~\eqref{ev34} describe the
distribution of the system only over a part of levels, rather
than over all levels. Note also that, if the $g_i$ satisfy the
condition
\begin{equation}
g_i=\left\{
\begin{array}{lcl}
Ae^{-\beta\lambda_i}&\mbox{for}&i\in I\\
0&\mbox{for}&i\notin I
\end{array}
,\right. \label{ev35}
\end{equation}
where  $A$ is an arbitrary positive number, then
$\wt{w}_i(n,g)$~\eqref{ev24} coincides with
$\rho_i$~\eqref{ev34} for any $n$. Thus, the full and partial
Gibbs distributions are invariant with respect to the evolution
described above.

\begin{remark} Consider the special case where the random variable
$\lambda$ takes the same values as the polynomial
\begin{equation}
\cE(N_1)=\sum_{q=0}^p A_q N_1^q, \label{56}
\end{equation}
where $N_1$ ranges over $0,1,\dots,N$. Consider the
satisfiability of condition~\eqref{ev1} in this case.

\begin{proposition}
If $p<N$, then condition~\eqref{ev1} holds for no values of the
coefficients $A_q$, and if $p\ge N$, then condition~\eqref{ev1}
holds for generic coefficients $A_q$.
\end{proposition}
\end{remark}

\newpage
\section{Zeroth-order phase transition}

\subsection{Zeroth-order phase transition: general considerations}

\subsubsection{Main large parameters in thermodynamics}

 Thermodynamics studies
steady-state processes in which, {\it independently\/} of its
initial state, the system comes to a state that remains the same
in what follows.
 But if such a state still varies,
then this is a thermodynamic variation only if the process occurs
{\it extremely\/} slowly~\cite{1}.
 In other words,
we slightly vary the state of the system and wait until it
returns to equilibrium and again is independent of the initial
state.

 Thus, a large parameter, time, is invisibly present
in this case, i.e., we observe the system at huge time intervals.
 On the other hand,
thermodynamics is the limit macroscopic theory obtained from the
microscopic statistical physics as the number of particles $N$
tends to infinity.
 Hence there are two large parameters,
and many things depend on their ratio.
 In turn, classical statistical physics
is the limit of quantum statistics as $h\to 0$ ($h$ is the Planck
parameter).
 Thus, already three large parameters ``meet" here:
time scale, number of particles, and~$1/h$.

 In thermodynamics,
sufficiently large time periods are also often considered, and in
these periods only a part of the system comes to equilibrium;
this is the so-called ``local equilibrium."
 For example, when plasma is heated by a magnetic field,
the ions begin to obey the Maxwell distribution after a large
observation period, and the entire system comes to the
thermodynamic equilibrium only after a significantly large period
of time.

\subsubsection{The Gibbs formula}

 Let us consider the Gibbs formula,
which has the following form in quantum statistics.

 Suppose that a system is characterized by
the Hamiltonian operator~$\wh H$ in a Hilbert space (in
particular, in the Fock space)~$\Phi$.
 Suppose that the operator~$\wh H$
has a nonnegative discrete spectrum
$\lambda_0,\lambda_1,\dots,\lambda_n\dots$\,.
 Then the free energy depending on
the temperature~$\theta$ is determined as
$$
E=\theta\ln\sum_{n=0}^\infty e^{-\lambda_n/\theta}\delta_n,
$$
where~$\delta_n$ is the multiplicity of the
eigenvalue~$\lambda_n$.

 We note that one more parameter appears in the Gibbs postulate,
namely,
$$
E=\theta\ln\lim_{M\to\infty}
\sum_{n=0}^Me^{-\lambda_n/\theta}\delta_n.
$$
 This is one of the most important facts,
since $\lim_{M\to\infty}$ and~$\lim_{N\to\infty}$ do not commute!
 It turns out that
one must first pass to the limit as $M\to \infty$ and then as
$N\to \infty$.

 If we speak about local equilibrium,
then, as a rule, the Gibbs formula deals with a subset
$\{\lambda_{n'}\}$ of the set $\{\lambda_n\}$.
 This must be done in the study
of the above problem concerning the ion distribution in plasma on
heating.
 This is also done in glass and
in several other physical problems.

\subsubsection{``Friction'' operator}

At temperature $\theta=   0$, the Gibbs formula gives the value
of the lowest eigenvalue, which, in physics, is called the
``ground state.''

 This is a law which has not been stated accurately
and which was called the ``energetic efficiency" by
N.~N.~Bogolyubov and others.
 The lower energy is energetically more efficient.
 For example, if a system is perturbed
by an operator $\wh V\:(\Phi\to \Phi)$ and this operator is
sufficiently small, then the transition matrix from the
state~$\lambda_n$ to the state~$\lambda_m$ is determined by the
matrix element $(\Psi_n\wh V\Psi_m^*)$, where~$\Psi_n$
and~$\Psi_m$ are the eigenfunctions corresponding to~$\lambda_n$
and~$\lambda_m$, respectively (see~\cite{2}).

 The square of this matrix element
is the probability of transition from the state with
energy~$\lambda_n$ to the state with energy~$\lambda_m$.
 But if
$\lambda_n< \lambda_m$, then this transition is ``not
energetically efficient" and hence, as if from the viewpoint of
statistical physics and thermodynamics, is unrealizable, i.e.,
from the mathematical viewpoint, it must be set to be zero.
 Then the transition matrix is not a self-adjoint matrix
with zero entries above the diagonal.
 This means that friction is taken into account.
 A pendulum oscillates and stops in the end
if friction is taken into account, i.e., it comes to the ``ground
state."
 The ``friction" operator was considered
in more detail in the author's paper~\cite{3}.

\subsubsection{ Phase transitions}

 The derivative of
$E(\theta)$ with respect to~$\theta$ is called the entropy, and
when the entropy has a jump at a point~$\theta_0$, it is said
that a first-order phase transition occurs; when the second-order
derivative has a jump, it is said that a second-order phase
transition occurs, etc.

 But, of course, it is not the function itself that has a jump,
but the leading term of its asymptotics as $N\to \infty$.
 In experiments,
this ``jump" can sometimes be actually extended in time, but we
have agreed to consider variations in sufficiently large periods
of time.
 In these periods,
not only thermodynamic, but also dynamical processes can occur.
 We neglect them and consider only the time periods
in which the system comes to equilibrium with the ``energetic
efficiency" taken into account~\cite{4}.

 The author discovered the {\it zeroth-order phase transition\/}
both in the theory of superfluidity and superconductivity and in
economics (a stock price break-down, a default, etc.), and, quite
unexpectedly, it has turned out that, in view of the natural
axiomatics (see~\cite{5}), the zeroth-order phase transition is
related to quantum statistics and thermodynamics.
 This phase transition
has not been noticed by physicists, and, perhaps, it contradicts
their ideas that the free energy can be determined up to a
constant.

\subsubsection{Metastable state}

 We consider a simple example of
semiclassical approximation of the one-dimensional Schr\"odinger
equation
$$
-h^2y_n^{''}+u(x)y_n=\lambda_ny_n, \qquad y_n(x)\in L_2, \quad
h\ll1,
$$
where $u(x)=(x^2-1)^2+qx$ and $q>   0$ is a constant.

 In classical mechanics,
this equation gives the picture of ``two cups with a barrier
between them.''
 If a particle is at the bottom
of the right-hand cup with higher walls, then it can get into the
other cup with lower walls only if the barrier disappears.

 From the viewpoint of the Gibbs postulate,
as $\theta\to 0$, the particle must be in the deeper well.
 Nevertheless,
it is obvious that if $h\ll 1$, then the particle will stay in
the well with higher walls for a very long time.
 In this case,
the summation in the Gibbs formula must be performed over the
subset of eigenvalues corresponding to the well with higher walls
whose eigenfunctions tend to zero as $h\to0$ outside this well.
 Moreover,
the temperature~$\theta$ must not be too high, so that the
eigenvalues above the barrier do not play any role in the Gibbs
formula.

 Such a state at a local minimum
of the potential well $u(x)$ is an example of a metastable state.

 If we consider the matrix element
of the transition from the lower level corresponding to the well
with higher walls to the lowest level~$\lambda_0$ (corresponding
to the bottom of the deep well), then it turns out to be
exponentially small with respect to~$h$, but any transition to
higher levels is forbidden according to the ``energetic
efficiency" law (i.e., we consider perturbations by ``friction").

\subsubsection{Superfluidity}

 N.~N.~Bogolyubov used Landau's ideas
(see the footnote on p.~219 in~\cite{6}) to show that
superfluidity is not a motion of the fluid, not a dynamics, but a
state of the fluid such as, for example, ice or vapor for water.

 This state corresponds to a metastable state of the system
such that any transition from this state to the normal state is
almost forbidden: it is exponentially small as $N\to \infty$.

 N.~N.~Bogolyubov proved this rigorously
under the assumption that the system of Schr\"odinger equations
is periodic; in other words, the Schr\"odinger equations were
considered on a torus.
 The spectrum of superfluid velocities
(the energy levels corresponding to the relevant momenta) was
discrete.
 This readily distinguishes the state of superfluidity
from the hydrodynamics of fluids.
 In the limit as the torus radius tends to infinity,
the spectrum does not become, as usual, a continuous spectrum,
but becomes an everywhere dense point spectrum.

In the author's opinion, the fact that the spectrum is everywhere
dense pointwise in the limit can be easily explained from the
physical viewpoint.

Indeed, if the system is in a state with a superfluid
velocity~$v$, then its transition to a state with larger velocity
is forbidden by the ``energy state efficiency" law, and its
transition to a state with any lesser velocity is forbidden by
the fact that any decrease in the velocity contradicts the notion
of superfluidity.
 From the mathematical viewpoint,
this means that the larger the torus radius, the less
(exponentially less) is the transition from one state to another.

\subsubsection{Zeroth-order phase transition}

{\bf1.}
 The author explained the spouting effect
discovered in 1938 by J.~Allen and H.~Jones in which a superfluid
was ``flowing'' through a capillary of diameter $10^{-4}$\,cm (in
fact, this superfluid was at the superfluid level of a metastable
state).

 The author used the two-level model to show that,
at a point heated (by light) till the phase transition
temperature~$\theta_c$, the heat capacity becomes infinite, the
entropy has a jump, and the free energy decreases to its lower
value, i.e., to the point attained by the curve issuing from the
ground state heated to the temperature~$\theta_c$.
 This means that
a zeroth-order phase transition occurs.
 In the present paper,
we show that the same picture also appears in N.~N.~Bogolyubov's
model of a weakly nonideal Bose gas~\cite{8}.

 This phenomenon can be easily explained
if we assume that superfluidity is not a thermodynamic state, but
the motion of a fluid.
 At
$\theta= \theta_c$, the fluid becomes viscous and cannot
penetrate through a thin capillary.

 But the point is that this is not a motion, but a state,
and so we have a zeroth-order phase transition.

 The following question arises.
 What will happen to superfluidity
not in a capillary, but in a rather thick pipe?
 Where is the zeroth-order phase transition?
 At
$\theta= \theta_c$, the fluid passes from the superfluid state to
the state of the ordinary fluid with viscosity and begins to flow
according to the usual laws of hydrodynamics.
 After a while,
the motion stops, and the transition from the superfluid state to
the state of a fluid at rest is a thermodynamic transition.
 All the intermediate flow is hydrodynamics,
and must be neglected for our time scale.

{\bf 2.}
 Let us consider the example studied in Sec.~1.5 in more detail.
 We slowly vary the constant~$q$.
 We note that,
in thermodynamics with the field taken into account, there are
two more thermodynamic quantities: the field intensity and the
charge.
 Thus, a variation in~$q$
is a variation in a thermodynamic variable.

 We shall show that the passage of~$q$
through the zero point results in a zeroth-order phase transition.
 Indeed, resonance occurs at
$q=   0$: the eigenfunctions are already not concentrated in one
of the wells and the probability (the square of each
eigenfunction) is identically distributed over both wells.
 The number of eigenvalues is ``doubled''
and the Gibbs distribution has a jump.
 It is precisely here that
one can see the role of time.
 A certain amount of time is required for half the function,
decreasing outside the first well, to be ``pumped" into the
second well.
 But the Gibbs formula does not take this into account.
 It may happen that we will have to wait
for this transition for a long time, as it was in the preceding
example with a thick pipe.

 If~$q$
becomes negative, then all the eigenfunctions remain both in the
first and in the second well, and the Gibbs formula is taken over
all the eigenvalues, rather than over some of them.

{\bf3.}
 Finally, we consider the effect of transition
into the turbulent flow for fluid helium, which, in fact, is very
close to Landau's idea concerning energy pumping between large
and small vortices.

 As will be shown in another paper,
it is precisely the resonance that occurs between vortices of
these two types that results in the zeroth-order phase transition,
which sharply changes the thermodynamic parameters from the
viewpoint of thermodynamics in which we take into account the
large scales of time between transitions and, naturally, the
averaging over these times.

{\bf4.}
 Since only the values of bank notes are important,
while their serial numbers do not play any role, and the
interchange of two bank notes of the same denomination is not an
operation, bank notes obey the Bose statistics.
 The averaging of gains is a nonlinear operation,
as well as addition.
 As was shown above,
the only nonlinear ``arithmetics" (the semiring) that satisfies
the natural axioms of averaging for Bose particles (bank notes)
has the form $a\oplus b=\ln(e^a+ e^b)$.

 This lead to a formula of Gibbs type.
 A variation in
$\beta=1/T$ can be treated as a variation in the rouble
purchasing power caused, for example, by printing a lot of new
bank notes.
 After this, for a period of time,
the balance (equilibrium) is again established.

 Ordinary financial efficiency
plays the role of energetic efficiency.
 The zeroth-order phase transition
is either a default or a crisis~\cite{5, 9}.

\subsection{An exactly solvable model}

 First, we consider
the one-dimensional Schr\"odinger equation for a single particle
on the circle
\begin{equation}
\wh H\psi_k(x)=E\psi_k(x), \qquad \psi_k(x-L)=\psi_k(x), \label{1}
\end{equation}
where $\psi_k(x)$ is the wave function, $x$ takes values on the
circle, and~$\wh H$ is a differential operator of the form
\begin{equation}
\wh H=\epsilon(-i\hbar\partial/\partial x), \qquad \epsilon(z)\in
C^\infty. \label{2}
\end{equation}

Here~$\hbar$ is the Planck constant, $\psi_k(x)=\exp(ip_kx)$,
where $p_k=2\pi\hbar k/L$, $k$ is an arbitrary integer, and the
corresponding eigenvalues are equal to
\begin{equation}
E_k=\epsilon(p_k). \label{3}
\end{equation}
 We pass from the Hamiltonian function
$\epsilon(p)$ to a discrete function $\wt\epsilon(p)$ of the form
\begin{equation}
\wt\epsilon(p)=\epsilon(n\Delta p) \qquad \text{for}\quad \Delta
p\Bigl(n-\frac12\Bigr)\le p<\Delta p\Bigl(n+\frac12\Bigr),
\label{4}
\end{equation}
where~$n$ is an arbitrary integer and~$\Delta p$ is a positive
constant.
 Then the Hamiltonian~\eqref{2} changes appropriately,
and we denote the new Hamiltonian by~$\wh{\wt H}$.
 The eigenfunctions of this operator coincide
with~$\psi_k(x)$, and the eigenvalues~\eqref{3} become
\begin{equation}
\wt E_k=\wt\epsilon(p_k). \label{5}
\end{equation}

In what follows, we assume that the constant~$\Delta p$ takes the
form
\begin{equation}
\Delta p=\frac{2\pi\hbar G}L, \label{6}
\end{equation}

where~$G$ is a positive integer.

In this case, it follows from~\eqref{4} that the set of energy
levels~\eqref{5} is the set of $G$-fold degenerate energy levels
\begin{equation}
\lambda_n=\wt E_{Gn}=\epsilon(p_{Gn}). \label{7}
\end{equation}

 The Schr\"odinger equation
for~$N$ noninteracting particles has the form
\begin{equation}
\wh H_N\Psi(x_1,\dots,x_N)=E\Psi(x_1,\dots,x_N), \label{8}
\end{equation}
where $\Psi(x_1,\dots,x_N)$ is a symmetric function of the
variables $x_1,\dots,x_N$ (bosons).
 The Hamiltonian~$\wh H_N$
is given by the formula
\begin{equation}
\wh H_N=\sum_{j=1}^N\wh{\wt H}_j, \label{9}
\end{equation}
where $\wh{\wt H}_j$ is the Hamiltonian of the particle with
number~$j$, which has the form
\begin{equation}
\wh{\wt H}_j =\wt\epsilon\biggl(-i\hbar\frac\partial{\partial
x_j}\biggr). \label{10}
\end{equation}

The complete orthonormal system of symmetric eigenfunctions of
the Hamiltonian~\eqref{9} has the form
\begin{equation}
\Psi_{\{N\}}(x_1,\dots,x_N)
=\Symm_{x_1,\dots,x_N}\psi_{\{N\}}(x_1,\dots,x_N), \label{11}
\end{equation}
where~$\Symm_{x_1,\dots,x_N}$ is the symmetrization operator over
the variables $x_1,\dots,x_N$, \ $\{N\}$ is the set of
nonnegative integers~$N_k$, $k\in \bf Z $, satisfying the
condition
\begin{equation}
\sum_{k=-\infty}^\infty N_k=N, \label{12}
\end{equation}
and the function $\psi_{\{N\}}(x_1,\dots,x_N)$ is equal to
\begin{equation}
\psi_{\{N\}}(x_1,\dots,x_N) =\prod_{s=1}^N\psi_{k_s}(x_s).
\label{13}
\end{equation}
 Here the indices
$k_1,\dots,k_N$ are expressed in terms of the set $\{N\}$ using
the conditions
\begin{equation}
k_s\le k_{s+1} \quad \text{for all}\ 1\le s\le N-1, \qquad
\sum_{s=1}^N\delta_{kk_s}=N_k \quad \text{for all}\ k\in\bf Z ,
\label{14}
\end{equation}
and~$\delta_{kk'}$ is the Kronecker symbol.
 The eigenvalues of the Hamiltonian~\eqref{9} are
\begin{equation}
E(\{N\})=\sum_{k=-\infty}^\infty\wt E_kN_k. \label{15}
\end{equation}
 We consider interparticle interactions
of the following type.
 We assume that the particles interact in pairs
and the interaction operator for particles with numbers~$j$
and~$k$ has the form
\begin{equation}
\wh V_{jk}=-\frac VNW\bigl(\wh{\wt H}_j-\wh{\wt H}_k\bigr),
\label{16}
\end{equation}
where $V>   0$ is the interaction parameter and the function
$W(\xi)$ is given by the formula
\begin{equation}
W(\xi)=\begin{cases}
1 &\qquad\text{for}\ |\xi|<D,
\\
0 &\qquad\text{for}\ |\xi|\ge D.
\end{cases}
\label{17}
\end{equation}

Here $D>   0$ is the parameter of the width of interaction with
respect to energy.
 The operator~\eqref{17} corresponds to the interaction
under which particles in a pair attract each other and radiate
the energy quantum $-V/N$ if the difference between their
energies is less than~$D$ and do not interact at all if the
difference between their energies is larger than~$D$.
 The Hamiltonian of the system of~$N$
particles with interaction~\eqref{17} has the form
\begin{equation}
\wh H_N=\sum_{j=1}^N\wh{\wt H}_j +\sum_{j=1}^N\sum_{k=j+1}^N\wh
V_{jk}. \label{18}
\end{equation}
 In view of~\eqref{16}, the sums in~\eqref{18} commute,
and hence the set of eigenfunctions of the Hamiltonian~\eqref{18}
coincides with~\eqref{11}.
 It also follows from~\eqref{16}
that the corresponding eigenvalues are equal to
\begin{equation}
\bf E  (\{N\}) =\sum_{k=-\infty}^\infty\wt E_kN_k -\frac
V{2N}\sum_{k=-\infty}^\infty \sum_{l=-\infty}^\infty W(\wt
E_k-\wt E_l)(N_kN_l-\delta_{kl}N_k). \label{19}
\end{equation}
 In what follows, we assume that
the interaction width is sufficiently small, satisfying the
condition
\begin{equation} D<\min_{n\ne m}|\lambda_n-\lambda_m|.
\label{20}
\end{equation}
 The set of energy levels~$\wt E_k$,
$k\in \bf Z $, coincides with the $G$-fold degenerate set of
levels~\eqref{7}.
 Hence, by~\eqref{20},
the set of energy levels~\eqref{19} of the system of~$N$ particles
under study can be written as
\begin{equation} {\bf E}  (\{\wt
N\}) =\sum_{n=-\infty}^\infty\lambda_n\wt N_n -\frac
V{2N}\sum_{n=-\infty}^\infty\wt N_n(\wt N_n-1), \label{21}
\end{equation}
where the level $\bf E  \{\wt N\})$ has the multiplicity
\begin{equation}
\Gamma(\{\wt N\}) =\prod_{n=-\infty}^\infty \frac{(G+\wt
N_n-1)!}{(G-1)!\,\wt N_n!}. \label{22}
\end{equation}
 Here
$\{\wt N\}$ denotes the set of nonnegative integers~$\wt N_n$,
$n\in \bf Z $, satisfying the condition
\begin{equation}
\sum_{n=-\infty}^\infty\wt N_n=N. \label{23}
\end{equation}

 Let us consider the partition function for
the system of~$N$ bosons with Hamiltonian~\eqref{18}.
 Since the energy levels and their multiplicities
are given by formulas~\eqref{21} and~\eqref{22}, respectively, the
partition function at temperature~$\theta$ takes the form
\begin{equation}
Z(\theta,N) =\sum_{\{\wt N\}}\Gamma(\{\wt N\}) \exp\bigl((-\bf E
\{\wt N\})\theta\bigr). \label{24}
\end{equation}
 Here the summation is performed over all sets
$\{\wt N\}$ with condition~\eqref{23} taken into account.

 In what follows, we assume that~$G$
depends on~$N$ and the following condition is satisfied:
\begin{equation}
\lim_{N\to\infty}\frac GN=g>0. \label{25}
\end{equation}

 By
$\wt F(\{\wt N\},\theta)$ we denote a function of the form
\begin{equation}
\wt F(\{\wt N\},\theta) =\bf E  (\{\wt
N\})-\theta\ln\bigl(\Gamma(\{\wt N\})\bigr), \label{26}
\end{equation}
and by $\{\wt N^0\}$ the set of nonnegative numbers~$\wt N_n^0$,
$n\in \bf Z $, for which the function~\eqref{26} is minimal under
condition~\eqref{23}.

 Now, we consider the problem of finding the minimal value
of the function~\eqref{26} under condition~\eqref{23}.
 In the limit as
$N\to \infty$ and under condition~\eqref{25}, the point of minimum
has the form
\begin{equation}
\wt N_n(\theta)=N\bigl(m_n(\theta)+ o(1)\bigr), \label{27}
\end{equation}
where $m_n(\theta)$, $n\in \bf Z $, is determined by the system of
equations
\begin{equation}
\lambda_n-Vm_n+\theta\ln\biggl(\frac{m_n}{g+m_n}\biggr)
=\mu(\theta), \qquad n\in\bf Z , \label{28}
\end{equation}
and $\mu(\theta)$ can be found from the equation
\begin{equation}
\sum_{n=-\infty}^\infty m_n=1. \label{29}
\end{equation}
 The substitution of~\eqref{27} into~\eqref{26}
and then the use of the asymptotic Stirling formula give the
following relation for the specific free energy:
\begin{align}
f(\theta) &\equiv\lim_{N\to\infty}f(\theta,N)
=\lim_{N\to\infty}\frac{\wt F(\{\wt N^0\},\theta)}N
\\ &
=\sum_{n=-\infty}^\infty\biggl(\lambda_nm_n-\frac V2m_n^2\biggr)
+\biggl(\theta m_n\ln\biggl(\frac{m_n}g\biggr)
-\theta(g+m_n)\ln\biggl(1+\frac{m_n}g\biggr)\biggr), \label{30}
\end{align}
where, for brevity, we omit the argument~$\theta$
of~$m_n(\theta)$.
 We introduce the notation
\begin{equation}
\omega_n=\lambda_n-Vm_n. \label{31}
\end{equation}

\remark{Remark}
 In the notation~\eqref{31},
the system of Eqs.~\eqref{28} and~\eqref{29} takes the form
\begin{equation}
\omega_n(\theta) =\lambda_n-V\frac
g{\exp((\omega_n-\mu)/\theta)-1}, \qquad n\in\bf Z , \label{32}
\end{equation}
\begin{equation}
\sum_{n=-\infty}^\infty\frac g{\exp((\omega_n-\mu)/\theta)-1}=1.
\label{33}
\end{equation}
 The system of Eqs.~\eqref{32} and~\eqref{33}
{\it exactly\/} coincides with the temperature Hartree
equations~\cite{10} for the system of~$N$ bosons with
Hamiltonian~\eqref{18}.
\endremark

 We shall study the solutions
of Eqs.~\eqref{28} and~\eqref{29}.
 For
$\theta=   0$, the system has many solutions, which we number by
the integer~$l$:
\begin{equation}
m_n^{(l)}=\delta_{ln}, \qquad n,l\in\bf Z . \label{34}
\end{equation}
 Among all the numbers~$l$,
we choose those that satisfy the condition
\begin{equation}
\nu_{nl}\equiv\lambda_n-\lambda_l+V>0 \qquad \text{for all}\quad
n\ne l. \label{35}
\end{equation}
 For these numbers, there exist solutions
of Eqs.~\eqref{28} and~\eqref{29} converging to~\eqref{34} as
$\theta\to 0$.
 The asymptotics of these solutions as
$\theta\to 0$ has the form
\begin{equation}
m_n^{(l)}\sim g\exp\biggl(-\frac{\nu_{nl}}\theta\biggr) \quad
\forall n\ne l, \qquad 1-m_l^{(l)}\sim g\sum_{n\ne
l}\exp\biggl(-\frac{\nu_{nl}}\theta\biggr). \label{36}
\end{equation}
 Thus, depending on the spectrum~$\lambda_n$,
$n\in \bf Z $, for sufficiently small values of the
temperature~$\theta$, the system of Eqs.~\eqref{28} and~\eqref{29}
has many solutions.
 These solutions,
along with the point of global minimum, also contain points of
local minimum of the function~\eqref{26}.
 The values of the function~\eqref{26}
at the points of local minimum are equal to the free energy of
metastable states.
 We consider the function~\eqref{26} at
$\theta=   0$.
 In this case, it coincides with the energy spectrum
of system~\eqref{21}.
 We consider the energy of the system
for the case in which almost all particles are at the energy
level~$\lambda_l$, which means that following conditions holds:
\begin{equation}
\wt N_n\ll N \qquad \forall n\ne l. \label{37}
\end{equation}
 Deriving~$\wt N_l$
from relation~\eqref{23} and substituting the result
into~\eqref{21}, we see that, in view of~\eqref{37}, the energy
spectrum of the system in this domain has the form
\begin{equation}
\bf E  (\wt N) \approx\lambda_lN-\frac{VN}2 +\sum_{n\ne
l}(\lambda_n-\lambda_l+V)\wt N_n. \label{38}
\end{equation}

 To the Hamiltonian~\eqref{18}, there correspond the Hartree
equation and the appropriate system of variational equations.
 To each
$l\in \bf Z $, there corresponds a solution of the Hartree
equation, and this solution describes the state
\begin{equation}
\wt N_n^{(l)}=N\delta_{nl}. \label{39}
\end{equation}
 Moreover,
the eigenvalues of the system of variational equations
corresponding to this solution of the Hartree equation coincide
with
\begin{equation}
\nu_{nl}=\lambda_n-\lambda_l+V, \qquad n\ne l. \label{40}
\end{equation}
 In~\cite{11, 12},
it was shown that if the eigenvalues of the system of variational
equations for the solution of the Hartree equation are real and
nonnegative, then such a solution corresponds to the ground state
or to a metastable state of the system.
 This means that to all~$l$
for which condition~\eqref{35} holds at $\theta=   0$, there
correspond metastable states of the system of bosons.
 As in the case of zero temperature,
to solutions of Eqs.~\eqref{28} and~\eqref{29} for $\theta\ne 0$,
there corresponds a temperature metastable state if the
point~\eqref{27} is a point of local minimum.
 Now, we note that,
for very large temperatures, the system of Eqs.~\eqref{28}
and~\eqref{29} has only one solution corresponding to the global
minimum of the function~\eqref{26}.
 The asymptotics of this solution as
$\theta\to \infty$ has the form
\begin{equation}
n_m(\theta) \sim g\frac{e^{-\lambda_m/\theta}}
{\sum_{l=-\infty}^\infty e^{-\lambda_l/\theta}}. \label{41}
\end{equation}
 The uniqueness of the solution at large temperatures
means that all metastable states disappear with increasing
temperature.
 The temperature at which a metastable state disappears
is critical for this state.

 Let us consider the behavior
of the entropy and the heat capacity of metastable states when we
approach the critical temperature.
 We consider the metastable state
to which there corresponds a solution of Eqs.~\eqref{28}
and~\eqref{29} converging to~\eqref{34} as $\theta\to 0$ for
some~$l$ for which~\eqref{35} holds.
 In what follows,
we assume that the function $\epsilon(p)$ satisfies the condition
$\epsilon(0)<\epsilon(p)$ for $p\ne 0$.
 Then, according to~\eqref{7},
we have $\lambda_0< \lambda_l$ for $l\ne 0$.
 Hence the solution of
Eqs.~\eqref{28} and~\eqref{29} converging to~\eqref{34} as
$\theta\to 0$ for $l=   0$ corresponds to the temperature ground
state of the system of~$N$ bosons.
 Moreover,
condition~\eqref{35} becomes equivalent to the condition
$\lambda_l-\lambda_0< V$.

 We assume that this inequality holds for
$l\ne 0$.
 The condition
that the corresponding solution $m_n^{(l)}(\theta)$ of
Eqs.~\eqref{28} and~\eqref{29} determines a point of local minimum
of the function~\eqref{26} under condition~\eqref{23} by
formula~\eqref{27} can be written as the following system of
inequalities:
\begin{equation}
\gathered \alpha_n^{(l)}(\theta) \equiv-V+\frac{\theta g}
{m_n^{(l)}(\theta)\bigl(g+m_n^{(l)}(\theta)\bigr)}>0 \quad
\forall n\ne l,
\\
\alpha_l^{(l)}(\theta) \equiv-V+\frac{\theta g}
{m_l^{(l)}(\theta)\bigl(g+m_l^{(l)}(\theta)\bigr)}<0, \qquad
-\sum_{n\ne l}\frac{\alpha_l^{(l)}(\theta)}
{\alpha_n^{(l)}(\theta)}<1.
\endgathered
\label{42}
\end{equation}

 We note that these inequalities hold for~\eqref{36}
as $\theta\to 0$.
 Inequalities~\eqref{42} follow from the condition
that the second variation of the function~\eqref{26} is positive;
the variation is calculated under condition~\eqref{23}.

 The metastable state disappears at the temperature
at which the last inequality in~\eqref{42} becomes an equality.
 We denote this critical temperature by~$\theta^{(l)}_c$.
 It follows from~\eqref{42} and Eqs.~\eqref{28}
and~\eqref{29} for $\theta< \theta_c^{(l)}$ that
$m_n^{(l)}(\theta)$ is an increasing function of the
variable~$\theta$ for $n\ne l$ and $m_l^{(l)}(\theta)$ is a
decreasing function of~$\theta$.
 Moreover, we see that
\begin{equation}
m_l^{(l)}(\theta)>m_n^{(l)}(\theta)>m_{n'}^{(l)}(\theta)
\end{equation}
if $\lambda_n< \lambda_{n'}$ and $n,n'\ne l$.

 From~\eqref{30} we obtain the following expression
for the specific entropy of a metastable state in the limit as
$N\to \infty$:
\begin{equation}
s^{(l)}(\theta)
=\sum_{n=-\infty}^\infty\biggl(\bigl(g+m_n^{(l)}\bigr)
\ln\biggl(1+\frac{m_n^{(l)}}g\biggr)
-m_n^{(l)}\ln\biggl(\frac{m_n^{(l)}}g\biggr)\biggr), \label{43}
\end{equation}
where, for brevity, we omit the argument~$\theta$ of
$m_n^{(l)}(\theta)$.
 Differentiating~\eqref{43}, we obtain
\begin{equation}
\frac{\partial s}{\partial\theta} =\sum_{n\ne l}\frac{\partial
m_n^{(l)}}{\partial\theta} \ln\biggl(\frac{g+m_n^{(l)}}{m_n^{(l)}}
\,\frac{m_l^{(l)}}{g+m_l^{(l)}}\biggr)
>0.
\label{44}
\end{equation}
 The last inequality follows from the properties
of $m_n^{(l)}(\theta)$.
 Since
the last inequality in~\eqref{42} becomes an equality at the
critical temperature, we can show that, as
$\theta\to\theta_c^{(l)}- 0$, the solutions of Eqs.~\eqref{28}
and~\eqref{29} behave as follows:
\begin{equation}
m_n^{(l)}(\theta)-m_n^{(l)}(\theta_c^{(l)})
\approx\frac{C^{(l)}}{\alpha_n^{(l)}(\theta_c^{(l)})}
\sqrt{\theta_c^{(l)}-\theta} \qquad \forall n\in\bf Z , \label{45}
\end{equation}
where~$C^{(l)}$ is a negative number.
 The substitution of~\eqref{45} into~\eqref{43}
shows that the derivative of the specific entropy with respect to
the temperature (this quantity is equal to the heat capacity
divided by the temperature) tends to infinity as
$\theta\to\theta_c^{(l)}- 0$ according to the law
$1/\sqrt{\theta_c^{(l)}-\theta}$.
 This means that the projection
of the Lagrangian manifold corresponding to the metastable state
on the $\theta$-axis is not uniquely determined in a neighborhood
of the critical temperature~\cite{4, 13}.
 The derivative of the temperature with respect to the entropy
vanishes as the critical temperature is approached.
 Therefore, as well as in view of~\eqref{44},
as was already pointed out, the Lagrangian manifold is uniquely
projected on the $s$-axis.
 We note that it follows from the properties of
$m_n^{(l)}(\theta)$ that the following conditions hold for
$\theta< \theta_c^{(l)}$:
\begin{equation}
\aligned m_0^{(l)}(\theta)<m_n^{(l)}(\theta) &\qquad \forall
n\ne0,l,
\\
\alpha_0^{(l)}(\theta)<\alpha_n^{(l)}(\theta) &\qquad \forall
n\ne0,l.
\endaligned
\label{46}
\end{equation}

\newpage

Inequalities~\eqref{46} mean that the potential barrier between
the energy levels~$\lambda_l$ and~$\lambda_0$ is less than the
energy barrier between~$\lambda_l$ and~$\lambda_n$ for $n\ne0,l$.
 This means that,
as the critical temperature is attained, the system of bosons
under study which is in the metastable state with number~$l$
changes its state in a jump and passes to the temperature ground
state.
 This is a zeroth-order phase transition,
since not only the entropy and heat capacity, but also the free
energy have jumps.

If the quantity
\begin{equation}
\min_{n\ne0}|\lambda_n-\lambda_0|= \delta
\end{equation} is small,
then the difference between the free energies in the zeroth-order
phase transition from the lowest metastable state to the ground
state is also small, and the heat capacity in this transition has
a singularity. The asymptotics of the partition function near the
critical point is given by the asymptotics of the canonical
operator in a neighborhood of a focal point. This asymptotics has
the form of an Airy type function of an imaginary argument and
has a singularity as $N\to \infty$. The parameters $\delta\ll 1$,
$N\gg 1$, and $L\gg 1$ can be chosen so that the form of the
phase transition point coincides exactly with the
$\lambda$-point~\cite{14}--\cite{24}.

If adequate measures increasing the Kolmogorov complexity of the
Hartley entropy are not taken, the financial system still does not
necessarily result in a zero-order phase transition (social
explosion).

Now, let us return to the problem about deposits in a pyramid and
banks, i.e., to the assumptions that the bills of the same value
with different numbers are identical. The phase transition
related to the disappearance of the condensate has the following
meaning in this case. As $\beta$ (price) decreases, the actions
cease to be bought and sold at some moment $\beta_0$, i.e.,
nobody trades them at $\beta < \beta_0$, although, seemingly, it
is more profitable to sell them at any price, and therefore,
somebody can speculate in resale, thereby reducing the price of
the actions to zero. Nevertheless, this paradoxical fact is
observed  in practice and mentioned in the literature.

Then the corresponding shares will be involved in a complicated
barter exchange, which makes for an increase in the entropy. The
same happens to other equities, in particular, to currencies
provided that the inflation rate is sufficiently high. All in all,
this ensures that the entropy and complexity increase and can
prevent zero-order phase transition, which occurs only at some
threshold temperature and for some nonlinear interaction swinging
the equilibrium system.

Now we consider the following model.
Assume that the number of people is $n=n_1+n_2$.
The culture (velocity, energy) level $\lambda_1$
typical of $n_1$ is twice less than the level $\lambda_2$
typical of~$n_2$.
The quantities $\lambda_1=1$, $\lambda_2=2$, and $1<\gamma<2$
determine the quadratic interaction
$$
\cE(N_1)= \left(N_1+2N_2-\gamma\frac{N_1^2}{2N}-\frac{\gamma
N_2^2}{2N}\right)+ M/N_0\ln
\frac{(N_1+n_1-1)!(N-N_1+n_2-1)!}{(n_1-1)!N_1!(n_2-1)!(N-N_1)!},
$$
where $n_1$ of people have $N_1$ units of money
and $n_2$ of people have $N_2$ units of money, and $N_1\ll N_2$.
For a sufficiently large $ M/N_0$
(an analog of temperature), there is a phase transition
to the state with $N_1=N$ and $N_2=0$
with an enormous outburst of the kinetic energy
(in physics, this is the Allen--Jones fountain effect~\cite{QuantumEc,3}).

\end{document}